\documentclass{article}
\usepackage[pdftex]{color,graphicx}
\usepackage{amsmath, amsfonts, amsthm, amssymb,latexsym,fancyhdr,mathrsfs,graphics,syntonly,transparent}
\usepackage[colorinlistoftodos]{todonotes}
\usepackage[colorlinks=true,linkcolor=black,citecolor=purple]{hyperref}
\usepackage[all,2cell]{xy}
\usepackage[utf8]{inputenc}

\newcommand{\Sp}[1]{\scriptstyle{(#1)}\displaystyle}
\newcommand{\Cat}[1]{\textbf{\underline{#1}}}
\newcommand{\Lind}[1]{\vphantom{)}_{#1}}

\newtheorem{Thm}{Theorem}[section]
\newtheorem{Lem}{Lemma}[section]
\newtheorem{Corr}{Corrollary}[section]
\newtheorem{Def}{Definition}[section]

\title{On Double Groups and the Poincaré group}

\author{Dany Majard}

\begin{document}
\maketitle

\begin{abstract}
 In \cite{Crane2007b}, Crane and Sheppard considered the structure of the Poincare group as a 2-Group, and derived important information about its representations in a 2-Category suited for representations of non-compact 2-groups, following a lead of \cite{Crane2007c}. In this paper, starting from the position that the most natural structure to describe cobordisms with corners, as in the recently published work of A. Voronov \cite{Feshbach2011}, is the cubical approach to higher category theory of Ehreshman, we explore some possibilities given by double groups to build TQFTs. Our main theorem is an extension of the work of \cite{Andruskiewitsch2009}, where we prove a theorem on the structure of maximally exclusive double groups. This result gives a presentation of the Poincare group where the distinction between boosts, rotations and translations is part of the structure, from which a TQFT could be build with space and spacetime transformations kept separate. This article drafts a program that will hopefully yield new state sum models of physical interest.
\end{abstract}
\tableofcontents

\section*{Introduction}
Since the seminal work \cite{springerlink:10.1007/BF02698547} of Atiyah, topological quantum field theories (TQFTs) have turned many non-mathematician's minds towards the categorical paradigm, contributing to the ever growing attention category theory has received. In addition, they put emphasis on the importance and significance of the monoidal structure a category might support. It slowly became apparent that what was dealt with was a simple case of a two dimensional category theory. Connections to knot theory appeared \cite{yetter2001functorial} and the link between category theory and geometry inspired and fueled the ranks of those who thought about climbing the categorical ladder proposed in \cite{Crane1994}, to what is now being labelled "Higher" category theory.\\
As the language of monoidal categories, braided monoidal categories, n-categories and eventually (n,r)-categories was formalized by  Street~\cite{gordon1995coherence,joyal1993braided,street1983enriched,street1996categorical}, Batanin \cite{batanin1996definition} and Baez/Dolan~\cite{baez1997introduction,Baez1995,baez1998categorification,baez2004higher,Baez2008} to cite a few, an idea of Ehreshmann~\cite{0162.32601} was mostly left behind. The idea is fairly simple and straightforward: considering cubical sets to replace graphs in the definition of categories. It can be done by a recursive process called internalization, which is presented in appendix A. It makes composition/pasting easy and it includes the globular sets as a degenerate case, though a rigorous proof is yet to appear in the weak setup. Amongst the authors that showed considerable interest in them are Brown~\cite{brown1997higher,brown1981algebra,brown1999double,brown1976double}, Paré and Dawson~\cite{R.Dawson-R.Pare1993}, Grandis~\cite{grandis1999limits,grandis2004adjoint} and more recently Andruskiewitsch and Natale  \cite{Andruskiewitsch2005}.
 These entities are called cubical categories, n-tuple categories or, as in  \cite{Feshbach2011}, n-fold categories. The simple way they "tile" allowed Brown to develop a non-abelian homotopy theory, replacing mappings of pointed $S^n$ to a topological space by mappings of n-cubes, leading to the higher homotopy Seifert-van Kampen Theorem ~\cite{brown2011nonabelian}. Andruskiewitsch and Natale found in the structure of double groupoids very interesting connections to the theory of extensions of groups \cite{Andruskiewitsch2005}, work that can be traced back to Lu and Weinstein \cite{lu1990poisson}.\\
Just as TQFTs paved the way to n-categories, they hinted towards the cubical case as well, take for example the work of  Kerler and Lyubashenko \cite{kerler2001non}, that used them to construct non-semisimple 3 dimensional TQFTs or the work of Morton \cite{morton2010cubical}. More recently Voronov and Feshbach used cubical categories to define extended TQFT on cobordisms with corners in all codimensions \cite{Feshbach2011}.\\
This article proposes some paths of research for cubical TQFTs by providing interesting examples of double groups, whose representation theory needs to be defined properly. As higher dimensional versions of groups embody some internal structures of groups, we hope to use them to keep these structures central to the models we build. This is not entirely new: in the last decade, interest has emerged to study 2-groups, or categorical groups. As these already showed interesting structures, Crane and Sheppard \cite{Crane2007b} described the Poincare group as a 2-group, and Crane worked with Yetter to build a proper category for their representations in \cite{Crane2007c}. In this paper we will take this idea further as we will show that a further decomposition of the Poincare group is embodied in a very particular double group and hence the representation theory will give richer behavior.\\
Another example follows the recent renewed interest in finite groups in theoretical Physics by the discovery of neutrino mixing and the definition of the "tribimaximal mixing" matrix \cite{HarrisonPhys.Lett.B530:1672002} that seems to govern the phenomenon. It is provided by finite subgroups of the 4-dimensional euclidian group. In \cite{Ma2007}, Ma shows how the tribimaximal matrix would emerge from a symmetry corresponding to even permutations of 4 elements, i.e. the group $A_4$, or possibly a bigger finite group sharing some common properties. His article opened the door to finite symmetries in the quest for a theory of everything and although the relevance of the example is far from understood, we think that it is worth mentioning. The fact is that in the 4 dimensional case, finite subgroups of the rotation group correspond to core diagrams of double groups. Working out this case might lead us to draw a bridge between TQFTs and Ma's theory.\\\\
The first section is a reminder of the definition of double categories. The reader who is already familiar with them may skip it and refer to it for notation. The second section is a reminder of the structure of groupoids. The formalism developed there is necessary to the understanding of double groupoids and therefore double groups. The third section is a study of the structure of double groupoids, mostly following results from Andruskiewitsch, culminating in a theorem relating a new class of double groups to certain decompositions of groups. The fourth and fifth section will use these results to put into new light the Poincare group and the finite subgroups of SO(4).\vspace{5mm}

\textbf{Acknowledgements :} The author would like to thank my PhD thesis advisor Louis Crane and David Yetter for their help and support, and B. Bischof for proof reading. This paper would not exist otherwise.

													        \section*{Notation}
\begin{itemize}
\item The notation \Cat{C}(a,b) will be used for the homset of the category \Cat{C} from an object a to an object b. When $a=b$ we will reduce the notation to $\Cat{C}(a)$ and we will write $\Cat{C}[a]$ for the group of invertible endomorphisms.
\item Composition will be written in diagramatic order, i.e. $fg$ will mean f first and then g, which is more classically written $g\circ f$.
\item For a category \Cat{C}, the set of objects will be denoted $\Cat{C}_0$. 
\item Identities will be drawn as segments, i.e. arrows without heads.
\item Let $f,g:p\to b$ be a cone over $t,s:b\to a$ and $\pi_b,\Lind{b}\pi:b_2\to\,b$ the pullback of $(t,s)$. Then $f\,\ulcorner g$ is the unique morphism $p\to b_2$ given by the universality of the pullback, i.e :
		\begin{align*}
		  (f\,\ulcorner g)\Lind{b}\pi &= f\\
		  (f\,\ulcorner g)\pi_b&=g
		\end{align*}
	Diagrammatically it give the following picture :
	\begin{figure}[!hbtp]
		\centering
		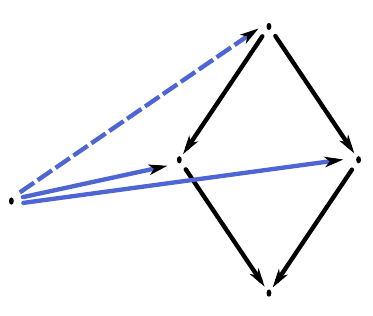
	\end{figure}\\
	$b_2$ can also be written $b\,\Lind{t}\!\times_s b$
\end{itemize}
													        \section{Double categories}
\begin{Def}A \textbf{double category} is a set $(O,H,V,S,s_{1,h},t_{1,h},s_{1,v},t_{1.v},s_{2,h},\\t_{2,h},s_{2,v},t_{2,v},\imath_{1,h},\imath_{1,v},\imath_{2,h},\imath_{2,v},\circ_{1,h},\circ_{1,v},\circ_{2,h},\circ_{2,v})$ where :
\begin{itemize}
	\item The sets O,H,V,S are respectively objects, horizontal arrows, vertical arrows and squares.
	\item The following maps of sets are the source and target maps :
		\begin{align*}
			s_{1,h},t_{1,h}\,:\,&H\to O
			&s_{2,h},t_{2,h}\,:\,&S\to V\\
                                s_{1,v},t_{1.v}\,:\,&V\to O
			&s_{2,v},t_{2,v}\,:\,&S\to H
		\end{align*}
	\item The following maps of sets  are the identity maps :
		\begin{align*}
			\imath_{1,h}&:O\to H
			&\imath_{2,h}&:V\to S\\
			\imath_{1,v}&:O\to V
			&\imath_{2,v}&:H\to S
		\end{align*}
	\item The following maps of sets are the composition maps.
		\begin{align*}
			\circ_{1,h}&:H\times_O H\to H
			&\circ_{2,h}&:S\times_V S\to S\\
			\circ_{1,v}&:V\times_O V\to V
			&\circ_{2,v}&:S\times_H S\to S
		\end{align*}
	\item The following holds :
	\begin{itemize}
		\item Sources and targets are compatible the following way :
			\begin{align*}
				s_{2,h}s_{1,v}&=s_{2,v}s_{1,h}
				&s_{2,h}t_{1,v}&=t_{2,v}s_{1,h}\\
				t_{2,h}s_{1,v}&=s_{2,v}t_{1,h}
				&t_{2,h}t_{1,v}&=t_{2,v}t_{1,h}
			\end{align*}
		\item $\imath_{i,\alpha}$ is an identity for the composition $\circ_{i,\alpha}$.
		\item the compositions  $\circ_{i,\alpha}$ are associative.
		\item  $\circ_{2,v}\circ_{2,h}=\circ_{2,h}\circ_{2,v}$ whenever it makes sense.
	\end{itemize}
\end{itemize}
\end{Def}
\clearpage
A general element of a double category looks like a square:
		\begin{figure}[!hbtp]
		\centering
		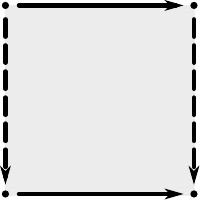
	\end{figure}\\
	It has vertical and horizontal sources and targets -the sides- which themselves have sources and targets, the corners. It also has horizontal and vertical identities and compositions. Composing is pictorially equivalent to pasting squares in one direction or the other. The above picture depicts the two different types of arrows with plain and dotted lines. It is meant to emphasize the fact that they are intrinsically different and do not compose with each other. It should be noted that the vertical source of a square is a horizontal arrow and vice-versa. Accordingly, vertical identities are identities on horizontal arrows. The last line of the definition is called the "interchange law". It ensures that any assortment of the sort :
	\begin{figure}[!hbtp]
		\centering
		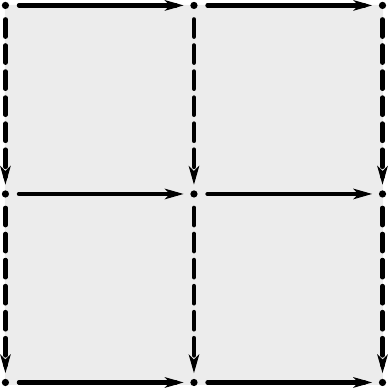
	\end{figure}\\
	yields the same square regardless of how it is composed. In a paper \cite{R.Dawson-R.Pare1993} published in 1993, Dawson \& Paré gave general conditions for the existance of composition for general ``tilings'', or arrangements of squares  and it is a good source for what happens in double categories. A consequence of the interchange law is that identities of identities are unique, i.e.:
\begin{align*}
	\imath_{1,v}\imath_{2,h}&=\imath_{1,h}\imath_{2,v}
\end{align*}
  There are a few equivalent definitions of double categories, including the original "squares only" definition of Ehresmann \cite{0162.32601}, while the most elegant is that they are internal categories in the category of categories. The above definition has the advantage of being more accessible for most readers.

	\begin{Def} A \textbf{double functor} between double categories is a map of sets sending objects to objects, arrows to arrows and squares to squares while respecting source, target, identities and composition.
	\end{Def}

	They are the two dimensional equivalent of functors of categories. Let's visualize what a functor $F$ would do :
	\begin{figure}[!hbtp]
		\centering
		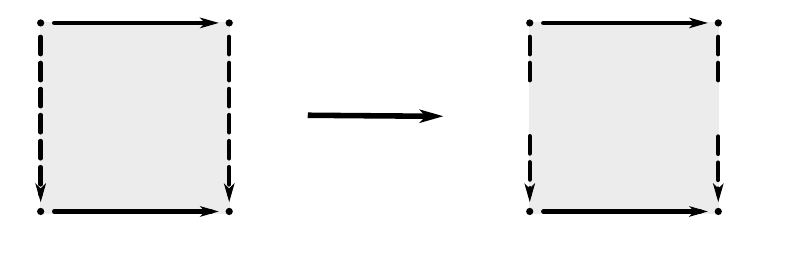
	\end{figure}	

A double functor reduces to a functor on the horizontal boundary category and the vertical boundary category. They compose in a unique way just as usual functors compose. This composition is associative and has identities.

	\begin{Def} Let F and G be double functors : $C\to D$, a \textbf{horizontal double natural transformation} from F to G is a map of sets that sends vertical arrows of C to squares in D in such a way that it intertwines these functors horizontally.
	\end{Def}
	Let's visualize a horizontal double natural transformation $\omega:F\to G$ :
	\begin{figure}[!hbtp]
		\centering
		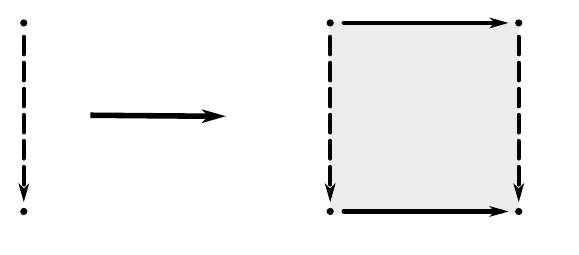
	\end{figure}\\
	The intertwining is represented by the following identity :
		\begin{figure}[!hbtp]
		\centering
		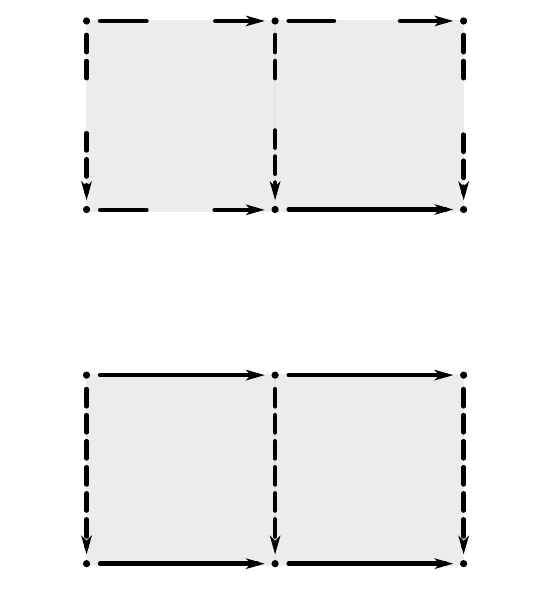
	\end{figure}
	\begin{Def}  Let F and G be double functors : $C\to D$, a \textbf{vertical double natural transformation} from F to G is a map of sets that sends horizontal arrows of C to squares in D in such a way that it intertwines these functors vertically.
	\end{Def}
	Let's visualize a vertical double natural transformation $\omega:F\to G$ :
	\begin{figure}[!hbtp]
		\centering
		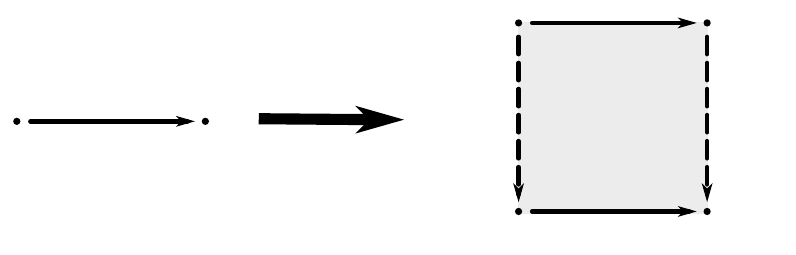
	\end{figure}\\
	The intertwining this time being vertical, it is represented by the following identity :
	\begin{figure}[!hbtp]
		\centering
		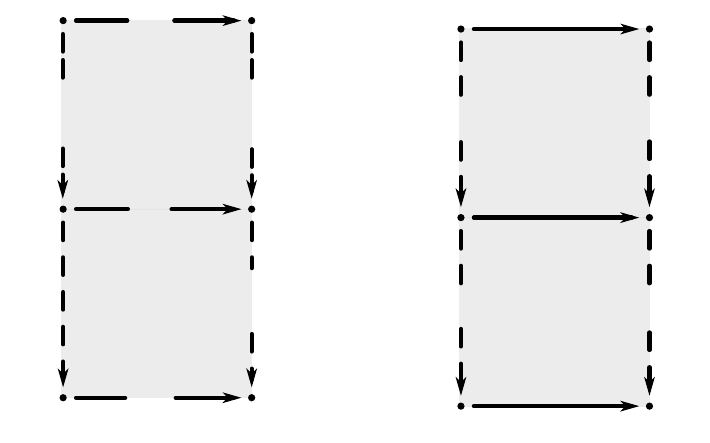
	\end{figure}\\
	From their intertwining properties it is clear that these double natural transformations compose in two ways and 
	that this composition is associative and unital. Being given horizontal and 
	vertical double natural transformations one could look for an entity intertwining them. Such an 
	entity exists and has been given different names throughout the literature, the most common being comparison and 
	modification.
	\begin{Def} Given four double functors $F,G,H,J:C\to D$, two horizontal double n.t. $\omega:F\to G$ and $\Omega:H\to 
	J$, and two vertical double n.t. $\delta:G\to H$ and $\Delta:F\to J$, a \textbf{double comparison} is a map of sets that sends objects of C to squares in D in such a way that it intertwines all the above entities in a plane.
	\end{Def}
	Let's visualize what such a double comparison does :
		\begin{figure}[!hbtp]
  		\centering
		  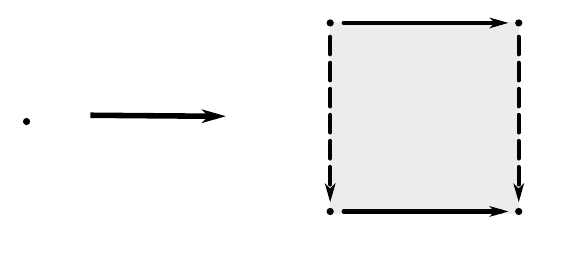
	  \end{figure}\\
	Intertwining on a plane means that for a given square in \Cat{C} the four compositions below are equal :\clearpage
		\begin{figure}[!hbtp]
  		\centering
		  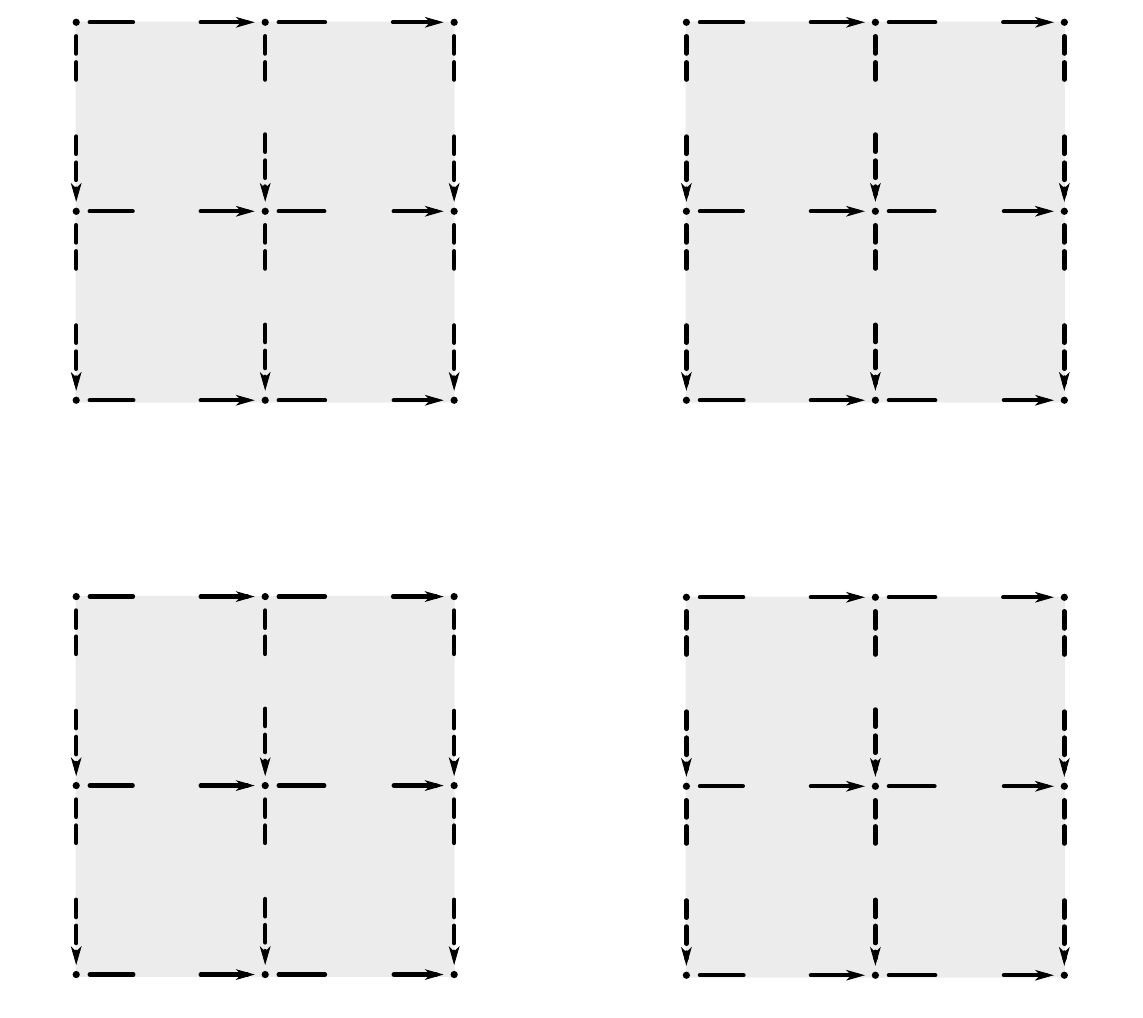
	  \end{figure}
												   	 \section{Structure of groupoids}
In this setion we present the structure of groupoids, as we developped it. It is in this spirit that we will approach higher groupoids as well.

\begin{Def}
	A \textbf{groupoid} is a category where all morphisms are invertible. It is \textbf{connected} if every homset is nonempty. It is \textbf{discrete} if all morphisms are identities.
\end{Def}
 With regards to this definition a group is a groupoid with one object and a group bundle is a totally disconnected groupoid.

\begin{Lem}Let \Cat{G} be a groupoid, then it is the disjoint union of connected groupoids
\end{Lem}
\begin{proof}Let \Cat{G} be a groupoid and suppose that $\Cat{G}(a,z) =\emptyset$. Then $\Cat{G}(a,b)\neq\emptyset$ implies $\Cat{G}(b,z)=\emptyset$, otherwise composition gives an arrow $a\to z$, giving a contradiction.
\end{proof}

\begin{Def}
	Let $X$ be a set, then the \textbf{coarse} groupoid \Cat{$\square X$} on X is the groupoid with X as objects and $\Cat{$\square X$}(a,b)=\{\vec{ab}\}$, i.e. every possible homset has a unique element.\\
	Let G be a groupoid, and $\Pi:\,\Cat{G}\to\square \Cat{G}_0$, be defined by 
	\begin{align*}
		\Pi(f)=\vec{ab}\qquad for\,f\in \Cat{G}(a,b)
	\end{align*}
	Its image is called the \textbf{frame} \Cat{$\blacksquare$G} of \Cat{G}. A groupoid is \textbf{slim} if it is isomorphic to its frame. The \textbf{core} \Cat{G$_\bullet$} of G is the subgroupoid of \Cat{G} consisting of arrows whose source and target are identical.
\end{Def}

In summary, the projection $\Pi$ collapses Homsets to only keep source and target, in other words it is a bundle :
\begin{figure}[!hbtp]
	 \centering
	 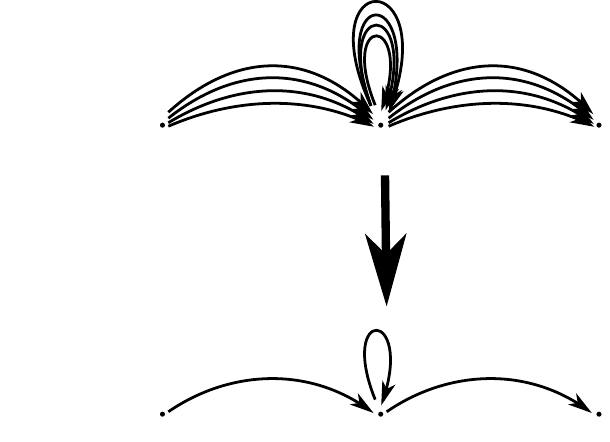
\end{figure}\\
The core of G, on the other hand, retains the information on what the "fibers" look like.

\begin{Lem}Let \Cat{G} be a connected groupoid, then its core is a principal bundle and its frame is coarse.
\end{Lem}
\begin{proof}Let $a,b\in\Cat{G}_0$, and pick $f\in\Cat{G}(a,b)$. Then $f^{-1}\circ(\_)\circ f :\Cat{G}(a)\to\Cat{G}(b)$ is an isomorphism of groups whose inverse is $f\circ(\_)\circ f^{-1}$. Since \Cat{G} is connected, $\Cat{G}(a)\simeq\Cat{G}(b)\quad\forall a,b\in\Cat{G}_0$,oving that its core is a principal bundle. Moreover, since $\Cat{G}(a,b)\neq\emptyset\,\forall a,b\in\Cat{G}_0$, the projection is on its frame is onto and full, i.e $\blacksquare \Cat{G}=\square\Cat{G}$.
\end{proof}


\begin{Thm}
	Let \Cat{G} be a groupoid, then every  section \footnote{graph morphism $!:\blacksquare G\to G$ that is the identity on objects} of $\Pi:\Cat{G}\to\blacksquare \Cat{G}$ in \Cat{Graph} defines a unique isomorphism of graphs:
	\begin{align*}
		\Phi:\,G_\bullet\,\Lind{t}\!\times_s\blacksquare \Cat{G}\to \Cat{G}
	\end{align*}
	Moreover there exists a unique groupoid structure on $G_\bullet\,\Lind{t}\times_s\blacksquare G$ that makes $\Phi$ 
	an isomorphism of groupoids.
\end{Thm}

\begin{proof}Assume a section $!:\blacksquare \Cat{G}\to \Cat{G}$ in \Cat{Graph}, then for $f\in \Cat{G}(a,b)$,  there exists a unique $u_f\in \Cat{G}(a,a)$ defined by $u_f:= f (!(\vec{ab}))^{-1}$ such that $f=u_f !(\vec{ab})$. \clearpage
\begin{figure}[!hbtp]
	 \centering
	 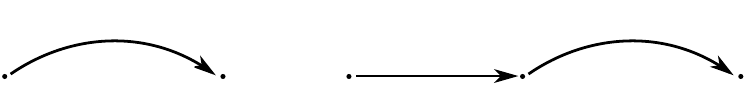
\end{figure}
This shows that $\Phi(f):=(u_f,\,\Vec{ab})$ defines an isomorphism of graphs and since it is the identity on objects, it lifts to an isomorphism on the pullback $\Cat{G}\,\Lind{t}\times_s \Cat{G}$ of the target and source map. Then the multiplication $m$ on \Cat{G} defines a multiplication $\#$ on $G_\bullet\,\Lind{t}\times_s\blacksquare G$ by $(\Phi\ulcorner\Phi)^{-1}m\Phi$. Similarly it defines an identity $\imath'$ by $\imath\Phi$. The picture is the following :
\begin{figure}[!hbtp]
	  \centering
	  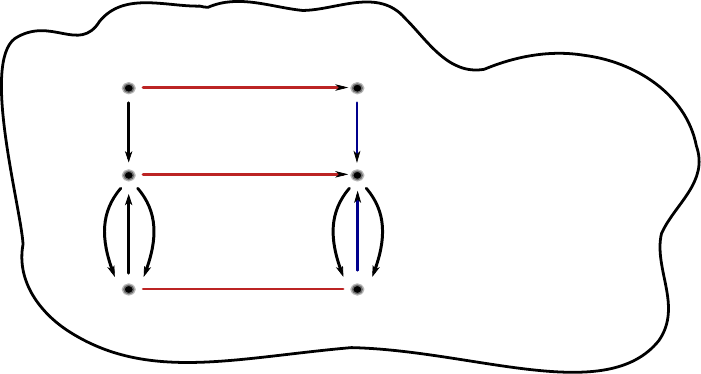
\end{figure}\\
The axioms for associativity and unit follow from $\Phi$ being an isomorphism of graphs.
\end{proof}

\begin{Corr}
	Let \Cat{G} be a connected groupoid with a chosen fiber bundle structure on its core with fiber G. Then every  section \footnote{graph morphism $!:\blacksquare G\to G$ that is the identity on objects} of $\Pi:\Cat{G}\to\square \Cat{G}_0$ in \Cat{Graph} defines a unique isomorphism of graphs:
	\begin{align*}
		\Phi:\,G\times\square \Cat{G}_0\to \Cat{G}
	\end{align*}
	Moreover there exists a unique groupoid structure on $G\times\square \Cat{G}_0$ that makes $\Phi$ 
	an isomorphism of groupoids.
\end{Corr}

Define the following notation : $u_{\vec{ab}}:=(u,\vec{ab})\in G_\bullet\,\Lind{t}\!\times_s\blacksquare \Cat{G}$.

\begin{Lem}The above defined composition is given as :
\begin{align*}
		u_{\vec{ab}}\#v_{\vec{bc}}=\Bigl(u\bigl(\vec{ab}\neg v\bigr) \phi_{\vec{ab},\vec{bc}}\Bigr)_{\vec{ac}}
\end{align*}
by two maps :
\begin{align*}
	\phi:\blacksquare \Cat{G}_0\,\Lind{t}\!\times_s \blacksquare \Cat{G}_0 \to G_\bullet
	\qquad\neg:\blacksquare \Cat{G}_0\,\Lind{t}\!\times_s G_\bullet \to G_\bullet
\end{align*}
such that $\vec{ab}\neg :\Cat{G}(b)\to\Cat{G}(a)$ is a group isomorphism and that the following identities are satisfied:
\begin{align*}
   s(\phi_{\vec{ab},\vec{bc}})&=a&
  \vec{ab}\neg (\vec{bc}\neg u)&=
  	\phi_{\vec{ab},\vec{bc}}(\vec{ac}\neg u)(\phi_{\vec{ab},\vec{bc}})^{-1}\\
 \phi_{\vec{aa},\vec{ab}}&=\phi_{\vec{aa},\vec{aa}}&
  (\vec{ab}\neg \phi_{\vec{bc},\vec{cd}})&=\phi_{\vec{ab},\vec{bc}}\phi_{\vec{ac},\vec{cd}}(\phi_{\vec{ab},\vec{bd}})^{-1}
\end{align*}
Moreover the identity is given by $\imath'(a)=((\phi_{\vec{aa},\vec{aa}})^{-1})_{\vec{aa}}$.
\end{Lem}
\clearpage
\begin{proof}$\phi$ is defined by : $e_{\vec{ab}}\#e_{\vec{bc}}=:(\phi_{\vec{ab},\vec{bc}})_{\vec{ac}}$,
or in pictures :
\begin{figure}[!hbtp]
	  \centering
	  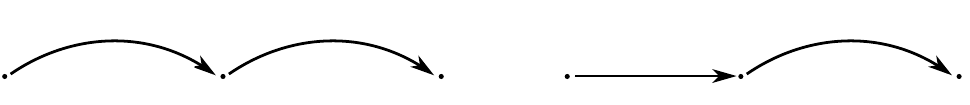
\end{figure}\\
It encodes how much the section fails to be a section of groupoids.\\
$\neg$ is defined by : 
$u_{\vec{ab}}\#v_{\vec{bc}}=:\bigl(u(\vec{ab}\neg v)\bigr)_{\vec{ab}}\#e_{\vec{bc}}$, or in pictures :
\begin{figure}[!hbtp]
	  \centering
	  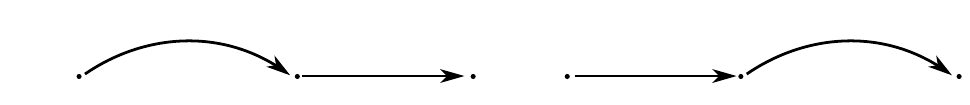
\end{figure}\\
or equivalently : $\vec{ab}\neg u:=(!\vec{ab})\, u\, (!\vec{ab})^{-1}$.
So that indeed, the above composition is given by $\Phi^{-1}$ of :
\begin{figure}[!hbtp]
	  \centering
	  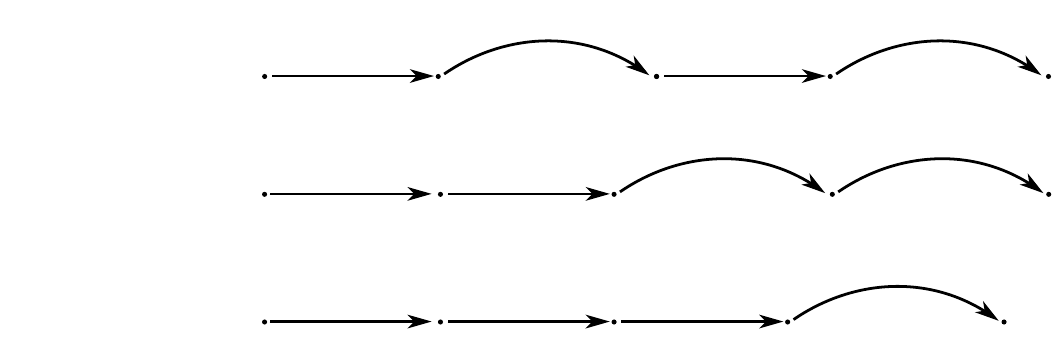
\end{figure}

Associativity in $G$ imposes the claimed relations on $\vartriangleright$ and $\phi$, via the following situations :
\begin{figure}[!hbtp]
	  \centering
	  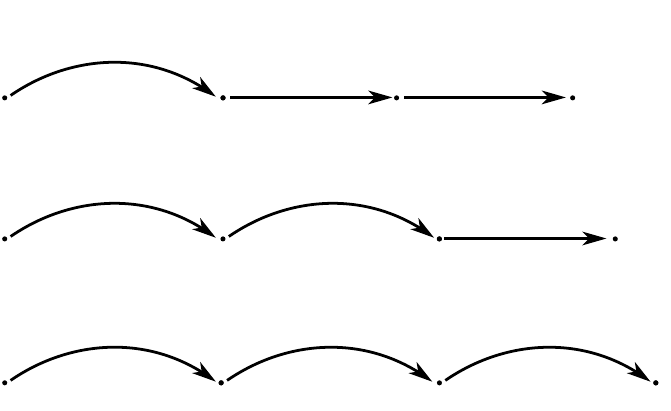
\end{figure}
Moreover $!(\Vec{aa})!(\Vec{ab})=\phi_{\vec{aa},\vec{ab}}!(\Vec{ab})$ shows that $\phi_{\vec{aa},\vec{ab}}=!(\Vec{aa})$ for all $b\in B$, hence the last condition. Now we can see that the identity $\imath'$ of $\#$ is given by $((\phi_{\vec{aa},\vec{aa}})^{-1})_{\vec{aa}}=\Phi^{-1}(id_a)$, via :
\begin{figure}[!hbtp]
	  \centering
	  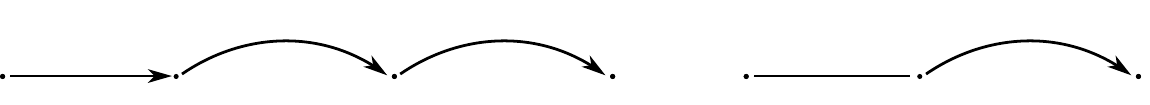
\end{figure}
\vspace{-0.5cm}\[\qedhere\]
\end{proof}

 This shows that necessary and sufficient information needed to describe a groupoid is given by its frame, its core, and two special maps.

\begin{Corr}Let $G$ be a group and B a set. Then connected groupoid structures on B with fiber G are given by pairs of maps 
\begin{align*}
	\phi:B\times B\times B \to G\\
	\rho:B\times B\times G \to G
\end{align*}
such that $\rho(a,b,\_)$ is a group isomorphism and that the following identities are satisfied:
\begin{align*}
  \rho(a,b,\rho(b,c,g))&=\phi(a,b,c)\rho(a,c,g)\phi(a,b,c)^{-1}\\
  \rho(a,b,\phi(b,cd))&=\phi(a,b,c)\phi(a,c,d)\phi(a,b,d)^{-1}\\
  \phi(a,a,b)&=\phi(a,a,a)
\end{align*}
\end{Corr}

\begin{Corr}Let $(B,G,\rho,\phi)$ and $(B,G,\tilde{\rho},\tilde{\phi})$ be two groupoids on a pair $(B,G)$ as defined 
	above. Then they are isomorphic if and only if there is a map $\Gamma:B\times B\to G$ satisfying :
	\begin{align*}
		\Gamma(a,b)\tilde{\rho}(a,b,g)&=\rho(a,b,g)\Gamma(a,b)\\
		\Gamma(a,b)\rho(a,b,\Gamma(b,c))\phi(a,b,c)
		&=\tilde{\phi}(a,b,c)\Gamma(a,c)
	\end{align*}
\end{Corr}
\begin{proof}
Suppose that you have two sections $!$ and $?$, then define $\Gamma_{\vec{ab}}:=?(\vec{ab})(!(\vec{ab}))^{-1}$. The above relations come from transforming $?$ to $!$ in the following pictures :
\begin{figure}[!hbtp]
	  \centering
	  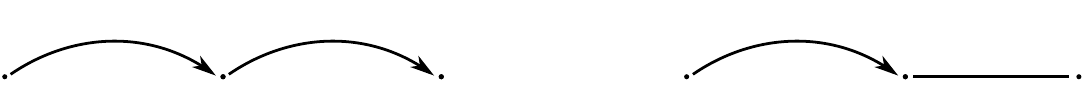
\end{figure}
\vspace{-5mm}
\[\qedhere\]
\end{proof}\vspace{-3mm}
  Remark that being given one such groupoid, all equivalent groupoid structures on $(B,G)$ can be constructed from it 
  by any map $\Gamma:B\times B\to G$. Without loss of generality, one can then consider 
  that $!\vec{aa}=id_a$ and the axiom $\phi_{\vec{aa},\vec{ab}}=\phi_{\vec{aa},\vec{ab}}$ becomes 
  $\phi_{\vec{aa},\vec{ab}}=id_a$.

														     \section{Double Groups}

 \begin{Def}A \textbf{double groupoid} is a (strict) double category where every square has both horizontal and vertical inverses\footnote{Note that consequently vertical and horizontal arrows are invertible as well, since their identities are.}. A \textbf{double group} is a double groupoid with a single object.
\end{Def}
To save time and space, we will take the convention not to represent objects any more unless absolutely necessary, and we will let the orientation of the page determine what is horizontal and what is vertical. The picture then becomes :\clearpage
	\begin{figure}[!hbtp]
		\centering
		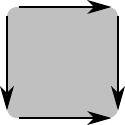
	\end{figure}
 The boundary groups of a double group are far from determining it. Further studies of double groupoids will tell us what extra structure they embody.
The notion of core gives way to different definitions in the case of double groupoids :
\begin{Def} Let $\tau$ be a double groupoid, then:
\begin{itemize}\addtolength{\itemsep}{-0.5\baselineskip}
\item	Its \textbf{core groupoid} $\tau_\lrcorner$ is the diagonal groupoid of elements of $\tau$ whose targets are identities. They are squares of the form :\vspace{-3mm}
	\begin{figure}[!hbtp]
	  \centering
	  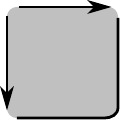
	\end{figure}\\	
	whose multiplication is defined by :
	\begin{figure}[!hbtp]
	  \centering
	  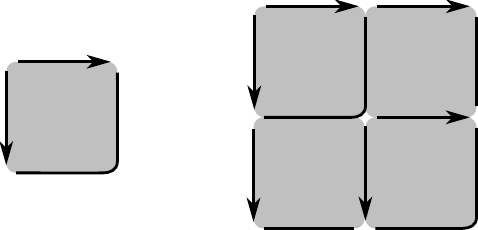
	\end{figure}
\item  Its \textbf{core bundle} $\tau_\bullet$ is the sub groupoid of $\tau_\lrcorner$ whose boundaries are all identities. It is a group bundle over the objects whose elements are squares are of the form :	
	\begin{figure}[!hbtp]
	  \centering
	  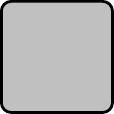
	\end{figure}
\item Its \textbf{core diagram} is the following diagram of groupoids:\vspace{-3mm}
	\begin{figure}[!hbtp]
	  \centering
	  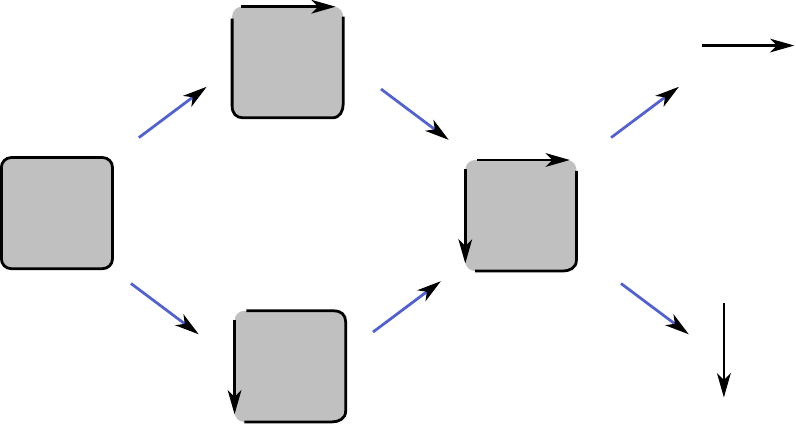
	\end{figure}
\end{itemize}
\end{Def}

\begin{Lem}The core groupoid is a groupoid.
\end{Lem}
\begin{proof}Identity and associtivity follow from the double category axioms. Let $X\in \tau_\lrcorner$ such that $f:=s_h(X)$, then $X^{-1}:=X^{-h}\circ_v id_f$. Then 
	\begin{figure}[!hbtp]
	  \centering
	  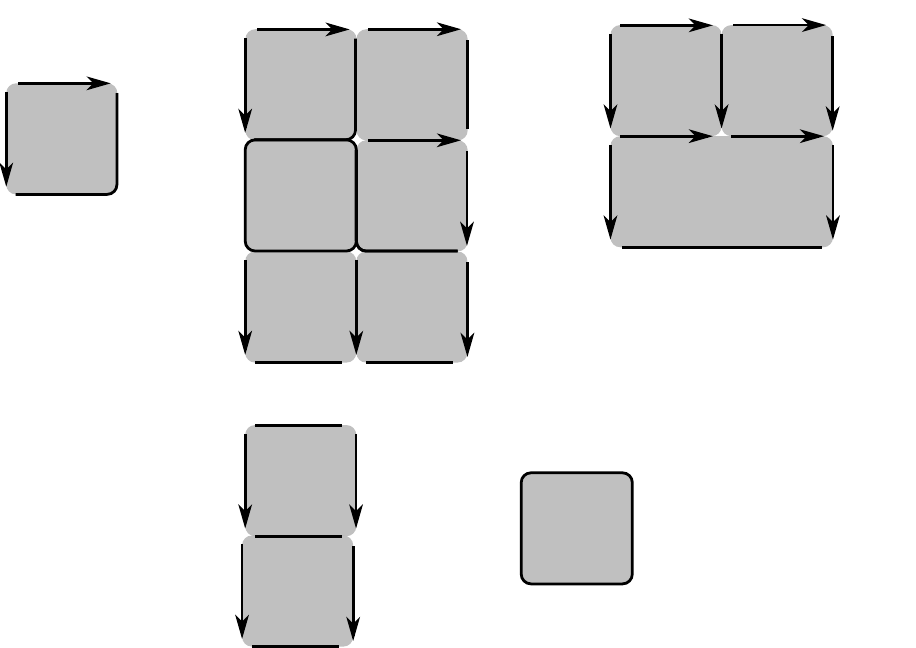\vspace{4mm}
	  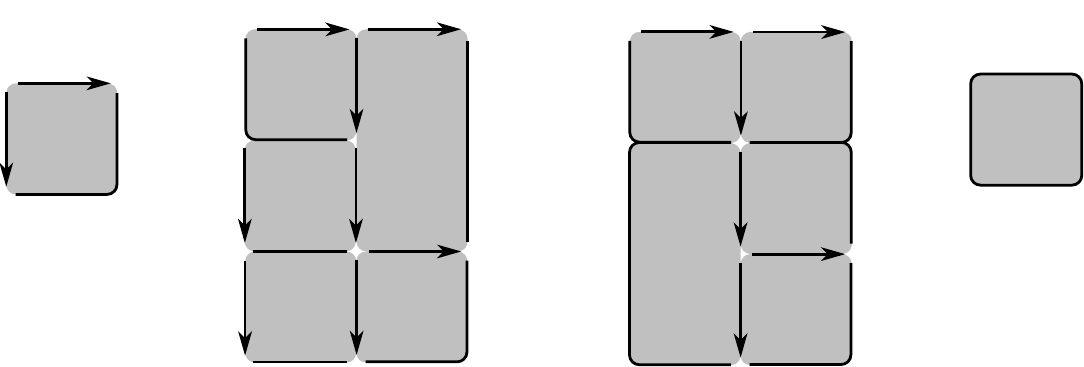
	\end{figure}\\
 proves that $X$ has both a right and a left inverse, and therefore is invertible.
\end{proof}

Core groupoids appeared in by Brown and Mackenzie's \cite{brown1992determination}, where they show that a certain class of double groupoids is determined by their cores. Note that it is not the case in general, though it provides a great deal of information.
Let's consider the simplest core diagram of a double group that is shared by two double groupoids:
\begin{align*}
	\xymatrix{
    &\{e\}\ar[dr]&&Z_2\\
    \{e\}\ar[ur]\ar[dr]&&\{e\}\ar[dr]\ar[ur]\\
    &\{e\}\ar[ur]&&Z_2\\}
\end{align*}
Then there are two double groupoids having it as a core groupoid. Since there is only one non-identity arrow in each direction, it will not be labelled. The smallest one has only identities :
\begin{figure}[!hbtp]
		\centering
		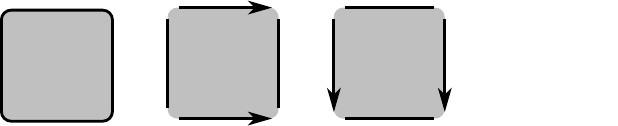
\end{figure}\\
The other one has one more square that is its own inverse :
\begin{figure}[!hbtp]
		\centering
		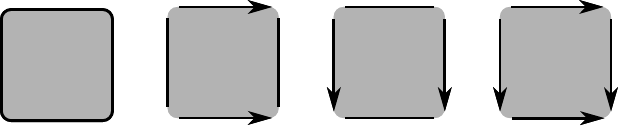
\end{figure}\\
Remark that any other square would impact the core diagram. The situation becomes more complex as the core diagrams involve bigger groupoids but even such a simple example makes the point. Let's analyse the core diagram further.

\begin{Lem}The core bundle of a double groupoid is an abelian group bundle over its objects.
\end{Lem}
\begin{proof}
This is the celebrated Eckmann-Hilton argument :
\begin{figure}[!hbtp]
		\centering
		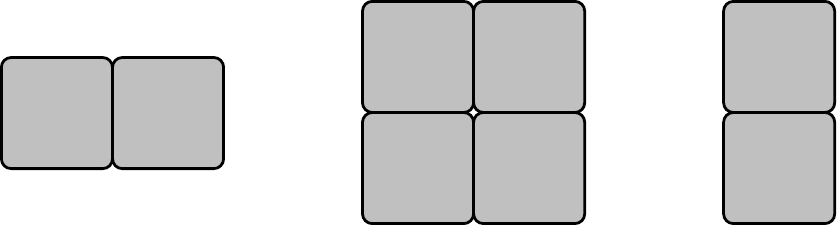
\end{figure}\\
This first step shows that the horizontal and vertical compositions are identical and correspond to the composition in the core bundle. If we continue to exploit the existance and uniqueness of identity squares on objects and the interchange law, we get the result announced.\\
\begin{figure}[!hbtp]
		\centering
		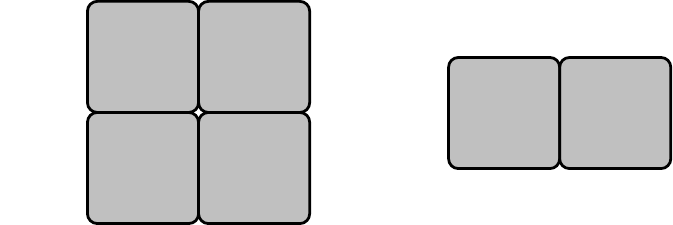
\end{figure}\\
Note that we are literally rotating A and B around each other.
\end{proof}
In \cite{Andruskiewitsch2005}, Andruskiewitsch and Natale show that  a special kind of double groupoid called vacant, corresponds exactly to factorizations of groupoids, which in the group case can also found in literature as matched pair of groups, bicrossed products, knit products or Zappa-Szep products \cite{Agore2009,Agore2010} of groups. A matched pair of groups is a pair of subgroups H,K of a group G such that each element of G can be uniquely written $g=hk$ for $h\in  H$ and $k\in K$. A matched pair of groups is equivalent to a bicrossed product of the same groups, written $H\Join K$. The rest of the section will extend the result to a bigger class of double groupoids. Let's recall the definition of vacancy.

\begin{Def}A double groupoid is \textbf{vacant} if any pair of possible horizontal sources and vertical targets is the boundary of a unique square.
\end{Def}

The following definitions will help broaden the notion:

\begin{Def}A double groupoid is \textbf{slim} if its core bundle is the trivial bundle. It is \textbf{exclusive} if its core groupoid and its core bundle are identical.
\end{Def}

Remember that the core bundle is a totally disconnected subgroupoid of the core groupoids, the one whose square have all boundaries identity arrows. Slim double groupoids have at most one cell per boundary, as the next lemma shows, so the only data they contain is which boundaries correspond to a square, in other words a slim double groupoid is a special subset of $H\times H\times V\times V$.

\begin{Lem}Let $\tau$ be a double groupoid, $X,Y\in \tau$ with the same boundary. Then there exists a unique element 
$u_{X,Y}$ in the core of $\tau$ such that :
	\begin{figure}[!hbtp]
	  \centering
	  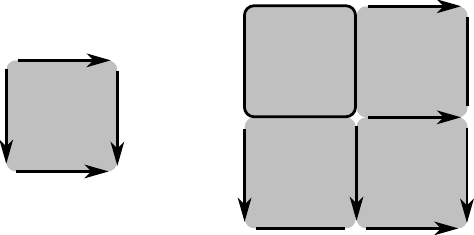
	\end{figure}
\end{Lem}
\begin{proof}
	Defining $u_{X,Y}$ by :
	\begin{figure}[!hbtp]
	  \centering
	  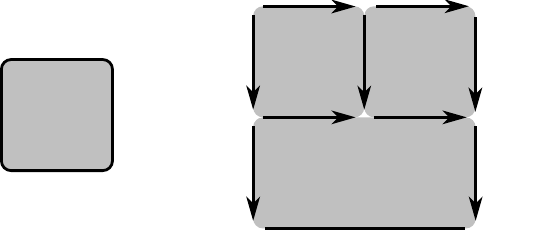
	\end{figure}\\
	and using inverses to isolate X yields the claim.
\end{proof}

\begin{Corr}A double groupoid is slim if and only if there is at most one square for a given set of boundaries.
\end{Corr}

	 Exclusive double groupoids may have many squares for a given boundary condition but this boundary is completely determined by the knowledge of two of its boundaries, one of each type, as the following lemma shows.\clearpage

\begin{Lem}Let $\tau$ be a double groupoid, $X,Y\in \tau$ with the same targets. Then there exists a unique element 
$t_{X,Y}$ in the core groupoid of $\tau$ such that :
	\begin{figure}[!hbtp]
	  \centering
	  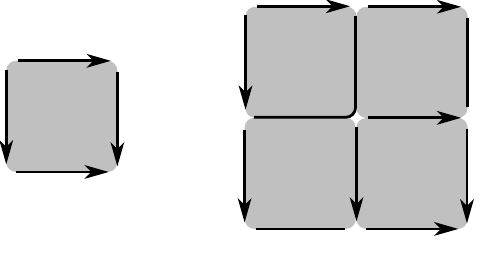
	\end{figure}
\end{Lem}

\begin{proof}Just as in the previous lemma, defining $t_{X,Y}$ by :
	\begin{figure}[!h]
	  \centering
	  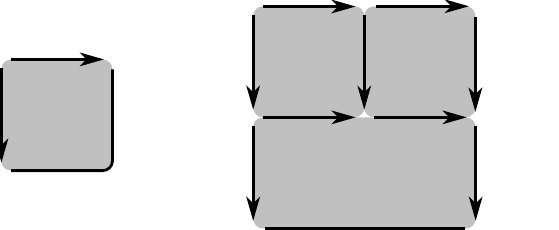
	\end{figure}\\
	and using inveses to isolate X yields the claim.
\end{proof}

\begin{Corr}A double groupoid is exclusive if and only if the boundary of its squares are determined by one of their boundaries of each type.
\end{Corr}
\begin{proof}Suppose two squares share one boundary of each type, then one of their inverses (horizontal, vertical or both) will share both their targets. By the above lemma they also share their sources. Using the same inverse again yields the result.
\end{proof}

\begin{Def}A double groupoid is \textbf{maximal} if any pair of vertical-horizontal arrows is the target pair of a square.
\end{Def}

In the example we gave earlier, both double groupoids were exclusive but only the second was maximal. Let's recall the definition of vacant

\begin{Lem}A vacant double groupoid is slim, maximal and exclusive.
\end{Lem}
\begin{proof}By the above lemma it is slim and exclusive. A simple use of vertical inverses show that it is maximal as well.
\end{proof}
A better definition for maximality would be the following:
\begin{center}"A double groupoid is maximal if it is maximal in the poset of double groupoids sharing a given core diagram."\end{center}
But it is yet to be proven that these definitions are equivalent in the case of exclusive double groupoids, which brings us to the problem of determining how many different exclusive double groupoids exist for a given core bundle. It is an open question and is reserved for further research. We will meet this problem again later in the paper.\\
On another note, it has been expected by the author that n-tuple groupoids that are vacant, i.e. maximal slim and exclusive, correspond to matched n-tuples of groups. It was conjectured by Brown in \cite{Brown11} and we prove it for all dimensions in \cite{Majard2}.
  Exclusive double groups are interesting in their own right and the next section will provide a prime example of these. The main result of this paper is the following theorem :

\begin{Thm}Maximal exclusive double groupoids equipped with a section are in one to one correspondence with fibered semi-direct products of an abelian group bundle with a matched pair of groupoids.
\end{Thm}

The proof, together with a more precisely stated theorem can be found in Appendix A. 

													        \section{The Poincare Group}

   To get to interesting examples, we chose to reduce our attention to double groups, in which case, the last chapter taught us that the structure of a vacant double group is a bicrossed product of groups. A very special case of such a decomposition is given by the Iwasawa decomposition of semisimple Lie groups, also called K(AN) decomposition for the subgroups appearing in the decomposition are usually denoted K, A and N. For more information on Lie double groups, see Brown and Macenzie \cite{brown1992determination}.\\
There is a Lie group that is very special to any theoretical physicist : the Lorentz group, SO(3,1). It is the group of symmetries of Minkowski spacetime in special relativity, i.e. the linear transformations of $\mathbb{R}^4$ preserving the diagonal matrix diag(-1 1 1 1) while fixing the origin. It contains rotations of euclidian space and "boosts", together with combinations of them, and while rotations form a subgroup, the boosts do not. Its K(AN) decomposition decomposes SO(3,1) accordingly, as shows the following lemma :
\begin{Lem}The Iwasawa decomposition of SO(3,1) is given by the following subgroups:
\begin{align*}
K&:=exp\Bigg\{\begin{bmatrix}0&0&0&0\\0&0&a&b\\0&-a&0&c\\0&-b&-c&0\end{bmatrix}\quad|a,b,c\in\mathbb{R}\Bigg\}\simeq SO(3)
\end{align*}
\begin{align*}
A&:=exp\Bigg\{\begin{bmatrix}0&a&0&0\\a&0&0&0\\0&0&0&0\\0&0&0&0\end{bmatrix}\quad|a\in\mathbb{R}\Bigg\}\simeq SO(1,1)\\
N&:=exp\Bigg\{\begin{bmatrix}0&0&a&b\\0&0&a&b\\a&-a&0&0\\b&-b&0&0\end{bmatrix}\quad|a,b\in\mathbb{R}\Bigg\}\\
\end{align*}
\end{Lem}
This result can be found in the book of J.Gallier \cite{gallier2011notes} amongst others.
 Keeping $AN$ as a single subgroup gives a bicrossed product presentation of the Lorentz group, which corresponds to a unique vacant double group. This vacant double group has boundary groups $SO(3)$ and $SO(1,1)\Join N$. \\
But just as the isometries of Euclidian space consist of rotations and translations, the Lorentz transformations form only a subgroup of the isometries of Minkowski spacetime: the Poincaré group. This group has then a canonical presentation as a semi-direct product between the Lorentz group and the group of translations in $\mathbb{R}^4$, which is abelian. We can then infer :

\begin{Corr} The Poincaré group has a decomposition of the form :
\begin{align*}Poinc\simeq (K\Join(AN)) \ltimes \mathbb{R}^4_+
\end{align*}
This decomposition is therefore represented by a maximal exclusive double group whose core is $(R^4,+)$ and boundary groups are $SO(3)$ and $SO(1,1)\Join N$
\end{Corr}

The corresponding maximal exclusive double group under the equivalence of the theorem has the following core diagram :
\begin{align*}
\xymatrix{
    &\mathbb{R}^4_+\ar[dr]&&SO(3)\\
    \mathbb{R}^4_+\ar[ur]\ar[dr]&&\mathbb{R}^4_+\ar[dr]\ar[ur]\\
    &\mathbb{R}^4_+\ar[ur]&&SO(1,1)\Join N}
\end{align*}
Schematically the roles of the different groups are summarized in the following picture :
	\begin{figure}[!hbtp]
	  \centering
	  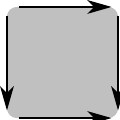
	\end{figure}\\
As claimed in the introduction, the presentation of the Poincare group as a double group has the very attractive property that the translations, the boosts and the rotations are kept distinct. Moreover, it is expected that maximal exclusive triple groups would allow to separate $SO(1,1)\Join N$. Of course, other Lie group decompositions of the Lorentz group can be considered and would give other double groups.

													        \section{Hyperplatonic solids}

 In our previous examples, we have considered very simple core groupoids, namely ones that were no bigger than the core bundle. But what about the remaining cases? Remark that, in the case of double groups, the core groupoid is, up to isomorphism, fully determined by the span of groups :
\begin{align*}\xymatrix{
    &\tau_{\lrcorner}\ar[dr]\ar[dl]\\
    \tau_h&&\tau_v}
\end{align*}
Indeed, taking the kernels of the projections and then the pullback of the obtained cospan restitutes the core diagram up to isomorphism. This brings us to the following characterization :
\begin{Lem}Core groups are in one to one correspondance with exact sequences
\begin{align*}
	\{e\}\to A\to G\to H\times K
\end{align*}
where $G,H,K\in\Cat{Grp}$ and $A\in\Cat{AbGrp}$.
\end{Lem}
\begin{proof} By universality of the pull-back that is given by the product, we have a unique arrow $\tau_\lrcorner\to\tau_h\times\tau_v$, as shown in the picture :
\begin{align*}
\xymatrix{\tau_\lrcorner\ar[dr]\ar[drrr]\ar@.[rr]^{!}&&\tau_h\times\tau_v\ar[dl]\ar[dr]&\\
	&\tau_h&&\tau_v}
\end{align*}
But the kernel of this map is the group of squares with horizontal and vertical sources trivial, i.e. the core bundle. This gives the exact sequence $\{e\}\to\tau_\bullet\to\tau_\lrcorner\to\tau_h\times\tau_v$ and proves the equivalence one way. For the other way, the projections of the product recover the span, from which we get the cospan on G by taking kernels. The span on A is recovered using the universality of the fibered product $P\times_GQ$, as shows the picture :
\begin{align*}\xymatrix{
	&&Q\ar[dr]&&&H\\
	P\times_GQ\ar[drr]\ar[urr]&A\ar@.[l]_{!}\ar[ur]\ar[dr]\ar[rr]&&G\ar[urr]\ar[drr]\ar[r]&H\times V\ar[ur]\ar[dr]\\
	&&P\ar[ur]&&&V}
\end{align*}
\end{proof}

\begin{Corr}Core groups of slim double groups are in one to one correspondence with exact sequences
\begin{align*}
	\{e\}\to\tau_\lrcorner\to\tau_h\times\tau_v
\end{align*}
\end{Corr}

In other words slim core diagrams correspond to subgroups of the product of the horizontal and vertical groups. This is particularly interesting in the light of the recent work of Ma\cite{Ma2007}. It seems that discrete symmetries may play an important role in upcoming theoretical physics.\\
In fact, regular and semi-regular polytopes in n dimensional Euclidian space can be studied from the finite subgroups of $SO(n)$, as shown in Coxeter's book \cite{coxeter1973regular}, which  for $n=4$ exhibits a behavior that doesn't exist in other dimensions. Indeed, the following lemma can be found in the literature, for example in \cite{varadarajan2004supersymmetry} .

\begin{Lem}There exist a short exact sequence :
	\begin{align*}\{e\}\to\mathbb{Z}_2\to Spin(n)\to SO(n)\to\{e\}
	\end{align*}
	Moreover, when $n=4$, there is an isomorphism $Spin(4)\simeq SU(2)\times SU(2)$.
\end{Lem}

This shows that subgroups of $SO(4)$ correspond to some subgroups of $SU(2)\times SU(2)$ by considering the preimage. The study of such subgroups was done by P.DuVal in \cite{du1964homo} where he shows that they are classified by group isomorphisms $A/A_0\to B/B_0$ where $A_0\vartriangleleft A\subset SU(2)$ and $B_0\vartriangleleft B\subset SU(2)$.
Putting the pieces together we understand that regular polytopes in dimension 4 give us, through their groups of symmetry, core diagrams of double groups.It is an important question to determine what or how many double groups share the same core diagram. We have managed to dodge this difficulty until now thanks to the availability of maximal double groups in the exclusive case but the question gets more complicated for arbitrary core diagrams. This done, we will be more prone to understanding which ones are of significant importance for the regular polytopes. Moreover it would classify double groupoids and is seen by the author as constituting a possible program for future research.

														  \section{Representations}

  As important to us as double groups are their representations. In dimension 3, the structure of the category of representations of certain algebraic entities, called quantum groups, gives us the basic ingredient to build a TQFT through state-sums. In the same fashion, some invariants of links and knots are given by categories of representations. Understanding the representations of our double groups will be important if we want to build TQFTs out of them. But first, let's recall what a representation of a group is. The basic notion of representation, called permutation representation, is a functor from the group G, viewed as a category with one object, to the category of sets :
\begin{align*}
G\to\Cat{Set}
\end{align*}
or in more classical terms a group morphism from G to \Cat{Set}[S], the group of automorphisms of S. A linear representation, on the other hand, is a functor to the category of vector spaces over a field $\mathbb{K}$, most often chosen as $\mathbb{R}$ or $\mathbb{C}$ :
\begin{align*}
G\to\Cat{$\mathbb{K}$-Mod}
\end{align*}
or in more classical terms it is a group morphism to the group of invertible linear maps on a vector space V, \Cat{$\mathbb{K}$-Mod}[V]. The study of representations of 2-groups has pushed people to consider specific higher dimensional analogs of vector spaces, such as \Cat{2-Vect} \cite{kapranov19942,barrett2006categorical,Baez2008}, or \Cat{Meas} \cite{Crane2007c}, while the basic notion of permutation representation was unsurprizingly thought of as functors to \Cat{Cat}, the prime example of a 2-category. Since 2-categories are special cases of double categories, representations of 2-groups should be special cases of representations of double groups and it therefore gives us suggestions to consider.

 It seems correct to think of permutation representations of double groups as functors from a double group G to \underline{\Cat{Cat}}, the "quintet" double category on the 2-category \Cat{Cat}, since it is the first double category at hand.
\begin{Def}Let \Cat{C} be a 2-category. Then the \textbf{quintet double category} of \Cat{C} is the double category whose :
	\begin{itemize}
		\item objects are the objects in \Cat{C}.
		\item horizontal arrows are the arrows of \Cat{C}
		\item vertical arrows are the arrows of \Cat{C}
		\item squares are 2-cells of $\Cat{C}:\, s_h(S)t_v(S) \to s_v(S)t_h(S)$  for a square S.
	\end{itemize}
\end{Def}
A square in a quintet double category then looks like the following :\clearpage
\begin{figure}[!hbtp]
	\centering
	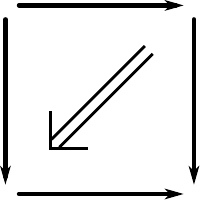
\end{figure}
with $\eta$ a natural transformation $fg\to hj$.\\
The quintet double category contains a quintet double groupoid just as a category contains its groupoid of invertible maps. For any object $a\in \Cat{C}$, one can then consider the double group \underline{\Cat{C}}[a], the restriction of the previous double groupoid to the object $a$. A representation of a double group would then be a double functor to $\underline{\Cat{C}}[a]$, for a given $a\in\Cat{C}$

Another possibility for a target of the representation functor is given by the fact that the homsets of the category of double categories are themselves double categories. Such representations of a double group would then contain a representation of the horizontal group as horizontal double natural transformations, a representation of the vertical group as vertical double natural transformations and squares as comparisons.\\
Both examples above have equivalents for $\mathbb{K}$-linear categories, k-linear functors etc, which would give notions of linear representations.
A third type of double category we may want to consider is the double category of functors and profunctors, cf \cite{grandis1999limits}. Simple examples must be taken to see the relevance of the above representation theories.

															 \section{Conclusion}
The hope that is induced by these examples of double groups is that their representations will be different enough from current theories that they will give us interesting new TQFTs. This program is in its infancy since the higher cubical category of 3-cobordisms with corners has not been described yet. In fact, it is what spurred this paper in the first place. Once it has been described, and the representation theory of double groups studied further, it may appear that they match in certain cases.\\

Another direction to look into is the "groupoidification" program of Baez \cite{Baez2009}. In this program spans of groupoids, which is what our core diagrams are, give linear maps between vector spaces that are determined by groupoids. It is an attempt to translate quantum mechanics to an algebraic theory with no mention of the continuum. It would be interesting to study the role of double groupoids in such a program, and it is our hope that future work will unveil new mathematics.
\clearpage
\appendix
														  \section{Internalization}
This first appendix will present the internalization process for categories and provide the proofs that double functors, natural transformations and comparisons have well defined compositions.
	\begin{Def}	An \textbf{internal category} in a category \Cat{C} with pullbacks is a sextuple $(a,b,s,t,\imath,\circ)$, where :
		\begin{itemize}
			\item a and b are objects of \Cat{C}.
		 	\item $(s,t,\circ,\imath)$ are arrows in \Cat{C} respectively called source, target, composition and identity as in the picture :
		 		\begin{figure}[!hbtp]
					\centering
					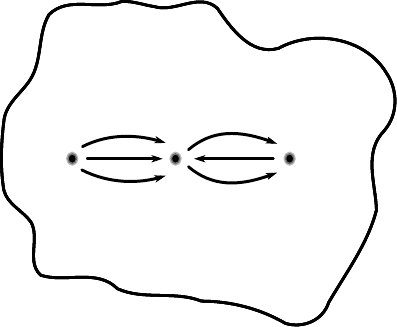
				\end{figure}\\
			where $b_2$, $\Lind{b}\pi$ and $\pi_b$ are the pull-back of $t$ and $s$.
				\begin{figure}[!hbtp]
					\centering
					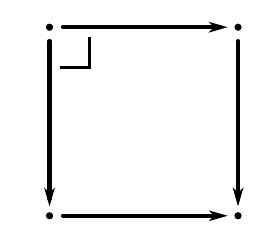
				\end{figure}\\
			More generally $b_n$ is a chosen limit of :
 				\begin{figure}[!hbtp]
					\centering
					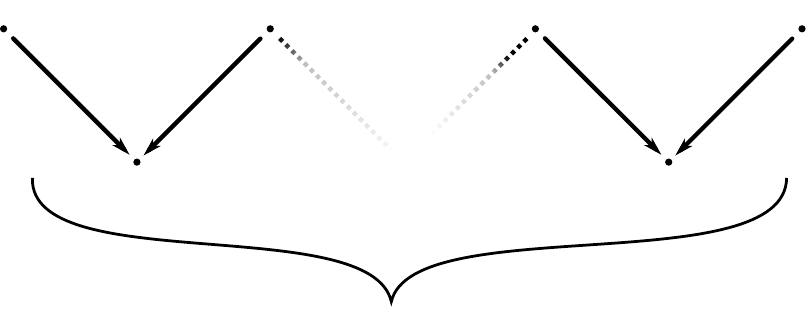
				\end{figure}
			\item some relations are satisfied:\clearpage
 				\begin{figure}[!hbtp]
					\centering
					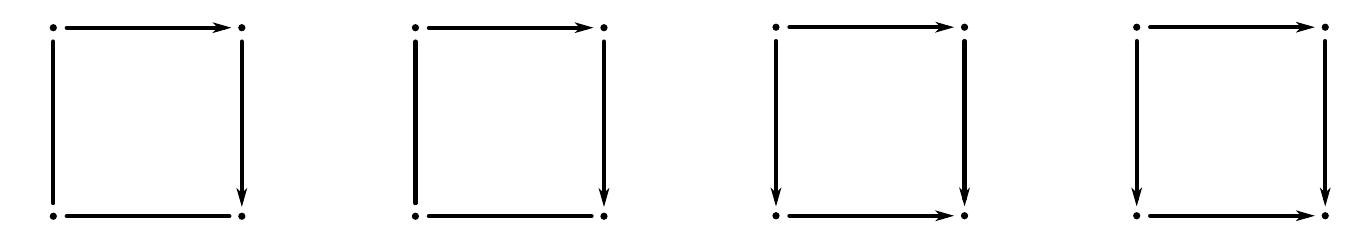
				\end{figure}
				\begin{figure}[!hbtp]
					\centering
					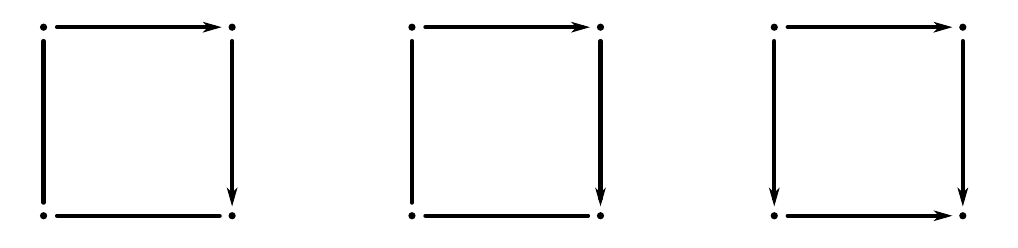
				\end{figure}
			\end{itemize}
	\end{Def}

 	This diagrammatic definition taken in \Cat{Set} yields the usual notion of a small category. Let's take a few lines to show why this is the case. The set $a$ is the set of objects of a category, the set $b$ its set of morphisms. The maps $s$ and $t$, sometimes denoted $d$ and $c$ for domain and codomain, associate respectively a source and a target to morphisms. The set $b_2$, pullback of the pair $(t,s)$, contains pairs of composable morphisms. The projections $\Lind{b}\pi$ and $\pi_b$ give respectively the first and second elements of the pairs and the map $\circ$ is the composition. The map $\imath$ selects an identity amongst morphisms with identical source and targets for every object. The axioms are then in order, the source and target axioms for identities and composition, the right and left unit axioms and finally the associativity axiom. Remark that the source and target axioms are a necessary condition for the unique maps used in the other axioms to exist.\\
 	We will soon consider the process of taking internal categories inductively, i.e. look at internal categories in the category of (internal categories in the category of(...(internal categories in \Cat{Set}))..)) but for now let's review a few facts about the construction. From now on the projections from the pullback will be omitted in the picture, and hence we draw an internal category in a category with pullbacks as :
 		\begin{figure}[!h]
		\centering
		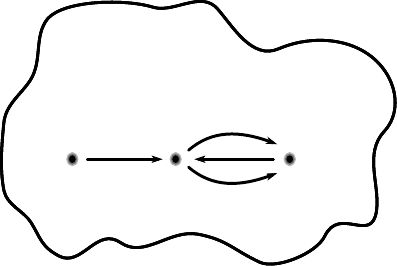
	\end{figure}\\
 	There are two ``degenerate'' versions of a category. The first one happens when the source or target morphism 
 	(s or t) is taken to be the identity. 
 	\begin{Lem}If any of $s,t,\imath$ are identities then they all are and so is the composition. The relations are then 
 	  trivially satisfied.\end{Lem}
 	\begin{proof}If $s=id$, then condition 1 forces $\imath=id$ and $t=id$ follows from condition 2. The projections 
 	from the pullback are then identities as well, which forces the composition to be the identity.
 	\end{proof}
 	\begin{Def}
 		An internal category whose defining morphisms are identities is called \textbf{discrete}.
 	\end{Def}
 	A discrete internal category in \Cat{C} does not bear any more structure than the objects of \Cat{C}. The second  degenerate case is the case where the object of objects, "a" in our definition, is a terminal object. The source and target morphisms are then unique and the structure becomes :
 		\begin{figure}[!hbtp]
		\centering
		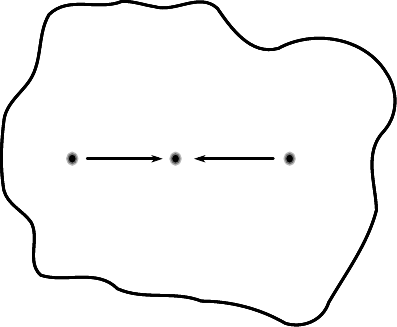
	\end{figure}
	\begin{Def}An internal category whose vertices object is terminal is called an \textbf{internal monoid}.
	\end{Def}

	Internal functors are sets of arrows that intertwine two instances of the structure, i.e. go from each object of the first 
 	instance to its equivalent in the second in such a way than all possible diagrams commute. 
 	\begin{Def} An \textbf{internal functor} between two internal categories is a triple of morphisms $(F_0,F_1,F_2)$:
 		\begin{figure}[!hbtp]
			\centering
			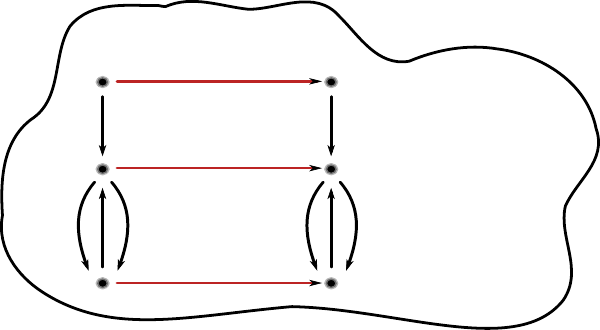
		\end{figure}\\
 		such that all possible diagrams commute, i.e.:\clearpage
 		\begin{figure}[!hbtp]
			\centering
			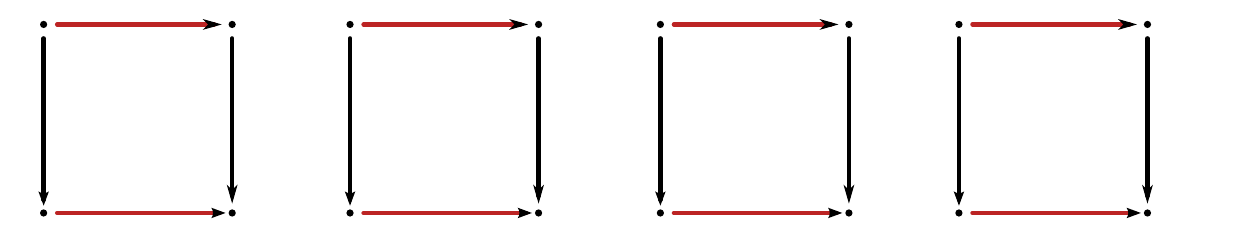
\vspace{1mm}
			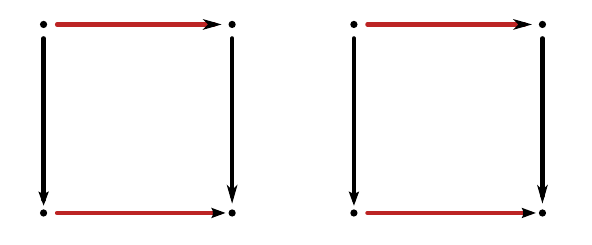
		\end{figure}
	\end{Def}
 		Note that since $b_2$ is a pullback, the first four conditions force $F_2$ to be the map factoring $\Lind{b}\pi 	
 		F_1$ and $\pi_bF_1$, as shown in the following picture :
 	 	\begin{figure}[!hbtp]
		  \centering
		  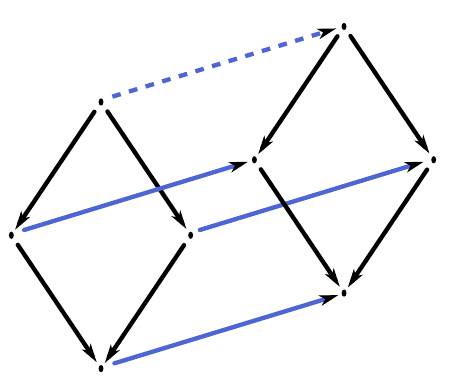
	  \end{figure}\\
Actually the definition should have a map between the chosen limits for any n but the relations would impose that they are the unique map given by the pullback. It will be understood from now on that the induced unique arrow is written $F_2:=\Lind{b}\pi F_1 \ulcorner \pi_b F_1$, or more generally $F_n:=\Lind{b}\pi F_1 \ulcorner\ldots\ulcorner \pi_b F_1$ n times, and we will give a functor by a pair of maps only $F:=(F_0,F_1)$.\newline
  Note that $(id_a,id_b)$ is always a functor and that functors have an obvious composition. Composing $F:=(F_0,F_1)$
  with $G:=(G_0,G_1)$ gives $F:=(F_0\,G_0,F_1\,G_1)$ and since it is reduced to the composition in the underlying 
  category it is therefore associative and with unit.\\
  Then, given a category \Cat{C}, one can build \Cat{IntCat(C)} whose objects are internal categories and arrows internal functors. In the case where \Cat{C} is \Cat{Set}, this category is simply \Cat{Cat}, the category of all small categories.\\
  As the real appeal for categories came with natural transformations, internal categories should have little to say 
  were we to stop here. The construction of natural transformations on categories and functor allows us to generalize 
  and define the notion internally.
  \begin{Def}Given two internal functors $F$ and $G$, an \textbf{internal natural transformation} is a morphism 
  $\omega$ :
  \begin{figure}[!hbtp]
		\centering
		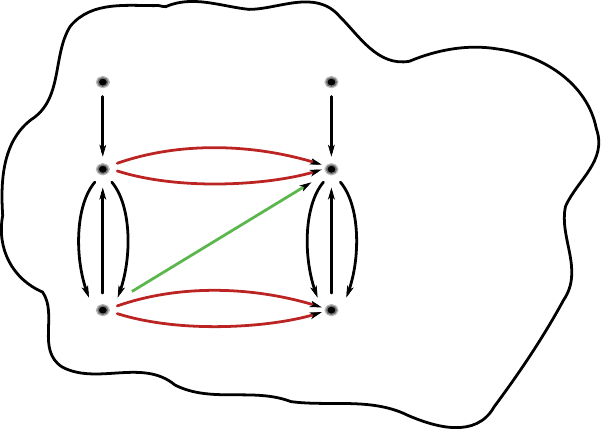
	\end{figure}\\
  such that the following diagrams commute:
  \begin{figure}[!hbtp]
		\centering
		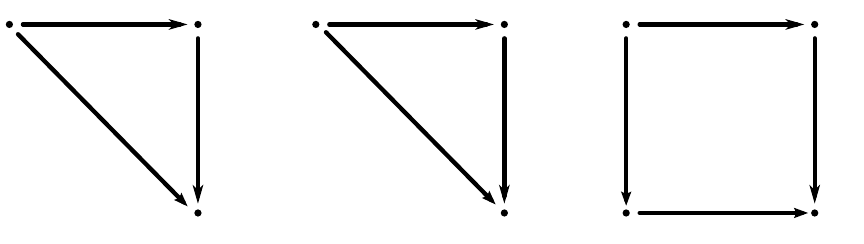
	\end{figure}
	\end{Def}
  Once again the first two conditions allow the unique maps used in the third one to exist. They are mere 
  prerequisites for the naturality axiom to be stated, and once again they concern the graph structure. For clarity, 
  let's draw the pullback diagram inducing $F_1\,\ulcorner t\omega$:
  \begin{figure}[!hbtp]
		\centering
		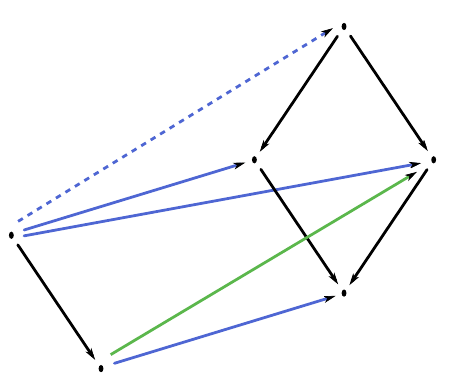
	\end{figure}\\
  Note that $F_0\imath$ is a natural transformation between $F$ and itself. Natural transformations compose 
  two different ways, usually called horizontal and vertical for reasons that should soon become obvious.
  \begin{Lem} Let $\omega:F\to G$ and $\gamma:G\to H$ be internal natural transformations, then $\omega \bullet 
  	\gamma:=(\omega\ulcorner\gamma) \circ'$ is an internal natural transformation $F\to H$. This composition is 
  	associative and has identities given by $F\to F_0\imath$.
  \end{Lem}
  \begin{proof}First let's show that it is a natural transformation from $F$ to $H$ :
  \begin{align*}
  	(\omega \bullet \gamma)s'
  	&=\bigl((\omega\ulcorner\gamma)\circ'\bigr)s'\\
  	&\bigl((\omega\ulcorner\gamma)\Lind{b'}\pi\bigr)s'\\
  	&=\omega s'=F_0
  	\\
  	(\omega \bullet \gamma)t'
  	&=\bigl((\omega\ulcorner\gamma)\circ'\bigr)t'\\
  	&\bigl((\omega\ulcorner\gamma)\pi_{b'}\bigr)t'\\
  	&=\gamma t'=H_0
  	\\
  	(F_1\ulcorner (t(\omega\bullet\gamma)))\circ'
  	&=(F_1\ulcorner (t(\omega\ulcorner\gamma)\circ'))\circ'\\
  	&=(F_1\ulcorner (t(\omega\ulcorner\gamma)))(1\ulcorner\circ')\circ'\\
  	&=(F_1\ulcorner t\omega \ulcorner t\gamma)(1\ulcorner\circ')\circ'\\
  	&=(F_1\ulcorner t\omega \ulcorner t\gamma)(\circ'\ulcorner 1)\circ'\\
  	&=(((F_1\ulcorner t\omega)\circ') \ulcorner t\gamma)\circ'\\
  	&=(((s\omega\ulcorner G_1)\circ') \ulcorner t\gamma)\circ'\\
  	&=(s\omega\ulcorner G_1\ulcorner t\gamma)(\circ'\ulcorner 1)\circ'\\
  	&=(s\omega\ulcorner G_1\ulcorner t\gamma)(1\ulcorner \circ')\circ'\\
  	&=(s\omega\ulcorner ((G_1\ulcorner t\gamma)\circ'))\circ'\\
  	&=(s\omega\ulcorner ((s\gamma\ulcorner H_1)\circ'))\circ'\\
  	&=(s\omega\ulcorner s\gamma\ulcorner H_1)(1\ulcorner\circ')\circ'\\
  	&=((s(\omega\ulcorner\gamma))\ulcorner H_1)(\circ'\ulcorner 1)\circ'\\
  	&=(s(\omega\bullet\gamma)\ulcorner H_1)\circ'
  \end{align*}
  Then from the uniqueness of the pullback, we find that $(F\ulcorner G \ulcorner H)((\Lind{b'_2}\pi\circ')\ulcorner 
  \pi_{b'})=((F\ulcorner G)\circ')\ulcorner H$, and then :
  \begin{align*}
  	(\omega\bullet \gamma)\bullet \lambda
  	&=(((\omega\ulcorner \gamma)\circ')\ulcorner \lambda)\circ'\\
  	&=(\omega\ulcorner \gamma \ulcorner \lambda)((\Lind{b'_2}\pi\circ')\ulcorner \pi_{b'})\circ'\\
  	&=(\omega\ulcorner \gamma \ulcorner \lambda)(\Lind{b'}\pi\ulcorner(\pi_{b'_2}\circ'))\circ'\\
  	&=(\omega\ulcorner((\gamma\ulcorner \lambda)\circ'))\circ'\\
  	&=\omega\bullet(\gamma\bullet \lambda)
  \end{align*}
  Therefore the uniqueness of the pullback and the associativity of composition insure the associativity of $\bullet$.
  Finally 
  \begin{align*}
  	(F_0\imath)\bullet \omega
  	&=(f_0\imath\ulcorner \omega)\circ'\\
  	&=((\omega s' i')\ulcorner \omega)\circ'\\
  	&=\omega(s'i'\ulcorner id_{b'})\circ'\\
  	&=\omega\\
  	\omega\bullet(G_0\imath)
  	&=(\omega\ulcorner G_0\imath)\circ'\\
  	&=\omega(id_{b'}\ulcorner t'i')\circ'\\
  	&=\omega
  \end{align*}
  So that as claimed $F_0\imath$ is the identity on $F$.
  \end{proof}
  Suppose now that you have the following situation :
  \begin{figure}[!hbtp]
		\centering
		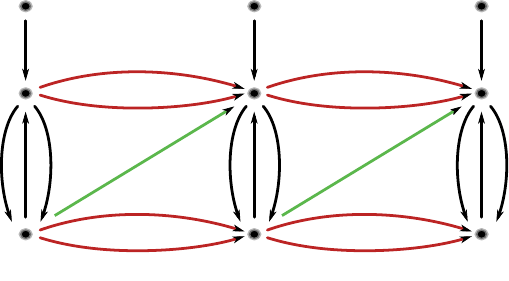
	\end{figure}\\
 	then one can create $\omega_H:=\omega H_1:\,FH\to GH$ and $\Lind{G}\theta:=G_0\theta:GH\to GL$\footnote{this is the 
 	whiskering}, which are natural transformations as well, and compose them with the previous defined 
 	composition, so that $\omega \diamond \theta := \omega_H \bullet \Lind{G}\theta$. 
 	\begin{Lem}The above constructed composition is associative and unital.
 	\end{Lem}
 	\begin{proof}
 		Let $\omega:\,F\to G$, $\theta:\,H\to L$ and $\eta:\,K\to M$, then :
	\begin{align*}
		(\omega\diamond\theta)\diamond\eta&=(\omega_H\bullet \Lind{G}\theta)_K\bullet\Lind{GL}\eta\\
		&=(((\omega H_1\urcorner G_0\theta)\circ' K_1)\ulcorner G_0L_0\eta)\circ''\\
		&=(((\omega H_1K_1\urcorner G_0\theta K_1)\circ'')\ulcorner G_0L_0\eta)\circ''\\
		&=(\omega H_1K_1\urcorner G_0\theta K_1\ulcorner G_0L_0\eta)(1\ulcorner \circ'')\circ''\\
		&=(\omega H_1K_1\urcorner G_0\theta K_1\ulcorner G_0L_0\eta)(\circ''\ulcorner 1)\circ''\\
		&=(\omega H_1K_1\urcorner G_0((\theta K_1\ulcorner L_0\eta)\circ''))\circ''\\
		&=\omega_{HK}\bullet\Lind{G}(\theta_K\bullet\Lind{L}\eta)\\
		&=\omega\diamond(\theta\diamond\eta)\\
		\omega\diamond \imath'&=\omega_{id}\bullet\Lind{G}\imath'\\
		&=(\omega\ulcorner G_0\imath')\circ'\\
		&=(\omega\ulcorner\omega t\imath')\circ'\\
		&=\omega(id\ulcorner t\imath')\circ'\\
		&=\omega
	\end{align*}
The left unit proof follows the same pattern as the right unit one.
 	\end{proof}
 	
	Internal categories in the category of categories are double categories, internal functors are double functors and internal natural transformations are horizontal double natural transformations. The fact that the category of categories is a 2-category is what gives us the extra structure of vertical natural transformation and comparisons. Vertical natural transformations are entities that share the same axioms as internal functors but where arrows are replaced by 2-cells. Comparisons share the same axioms as internal natural transformations but, once again, with 2-cells in place of arrows.

														\section{Proof of Theorem}

Here we will prove the theorem of Section 3. We will start with a definition.
\begin{Def}
	Let \Cat{2-sub} be the category whose objects are ordered triples $(G,H,K)$ consisting of a groupoid G and two of its subgroupoids H and K, and whose arrows $f:\, (G,H,K)\to (G',H',K')$ are functors $f:\,G\to G'$ such that $f(H)\subset H'$ and $f(K)\subset K'$.
\end{Def}
\begin{Def}
	Let $(G,H,K)\in\Cat{2-Sub}$, and define $\Gamma (G,H,K)$  as the double groupoid defined by :
	\begin{itemize}
		\item The horizontal arrow groupoid is H.
		\item The vertical arrow groupoid is V.
		\item There exist a square with sources $(h,k)$ and targets $(h',k')$ if anf only if $hk'=kh'$
	\end{itemize}
Define $\Gamma(f)$ by $\Gamma(f)(hk'=kh')=\bigl(f(h)f(k')=f(k)f(h')\bigr)$
\end{Def}

\begin{Lem}$\Gamma$ is a functor $\Cat{2-Sub}\to\Cat{slimDblGpd}$.
\end{Lem}

\begin{proof}First let's prove that $\Gamma(G,H,K)$ is a double category. Consider the squares $hk'=kh'$ and $jk''=k'j'$, then the horizontal composition of the two is the square $(hj)k''=k(h'j')$, which is well defined since $hjk''=hk'j'=kh'j'$. The vertical composition is defined similarly by $(hk'=kh')\circ_v(h'k'''=k''h'')=h(k'k''')=(kk'')h''$ and associativity follows from associativity in G. Identities are given by $he=eh$ and $ek=ke$ where $e$ is the identity of G. Now since the targets of $\Gamma$ are slim, functors between them are defined by their values on the boundaries, which are given by functors. $\Gamma$ then sends a functor to its restrictions on the given subgroupoids, which is done functorially. 
\end{proof}

\begin{Lem}The following are true :
	\begin{itemize}
	\item $\Gamma(G,H,K)$ is exclusive if and only if $H\cap K$ is discrete.\\
	\item $\Gamma(G,H,K)$ is maximal if and only if $HK=KH$
	\end{itemize}
\end{Lem}

\begin{proof}
	The core groupoid of $\Gamma(G,H,K)$ is composed of squares $he=ke$, hence it is $H\cap K$. When the latter is discrete, the double groupoid is exclusive and vice versa, which proves the first claim. Now squares exist if they correspond to elements of $HK\cap KH$, then $hk\notin KH$ if ands only if there exist no square $hk=\cdots$. If there exist such a square for every pair $(h,k)$, the double groupoid is maximal.
\end{proof}

\begin{Def}
	Let \Cat{VacDblGpd} be the category of vacant double groupoids and \Cat{2-matchSub} the subcategory of \Cat{2-Sub} that satisfy the above conditions. Then $\Gamma : \Cat{2-matchSub}\to\Cat{VacDblGrp}$ has a right adjoint $\Lambda$
\end{Def}

\begin{proof}
	Let $\tau$ be a vacant double groupoid, $X,Y\in\tau$. Define $\tilde\tau$ to be the groupoid whose objects are objects of $\tau$ and arrows $X:\,a\to b$ are squares $X$ with $s_hs_v(X)=a$ and $t_ht_(X)=b$. Composition is given by the following : if composable, $XY$ is be the unique square in $\tau$ having a barycentric subdivision with X in the upper left and Y in the lower right. The existence of such a square is guaranteed by maximality and its uniqueness by exclusivity and slimness. The same argument guarantees associativity and identities are given by $\imath_a:=\imath_h(\imath_v(a))$.\\
$\Lambda(\tau):=(\tilde\tau,\imath_v(\tau_h),\imath_h(\tau_v))$ defines the functor on objects and $\Lambda(\tau)\in\Cat{2-matchSub}$ since the following is true in $\tau$ :
\begin{align*}	
\begin{bmatrix}\imath&X\\\imath&\imath\end{bmatrix}=\begin{bmatrix}\imath&\imath\\X&\imath\end{bmatrix}
\end{align*}
For a double functor $f:\tau\to \tau'$, the corresponding functor $\Lambda(f)$ is given by $\Lambda(f)(X)=f(x)$. Since double functors preserve identities, $\Lambda(f)\in\Cat{2-matchSub}$.\\
The adjunction is then given by the isomorphism 
\begin{align*}\phi:\Cat{VacDblGpd}(\tau,\Gamma(G))\to\Cat{2-matchSub}(\Lambda(\tau),(G,H,K))
\end{align*}
 such that if $F(X)=(hk'=kh')$ then $\phi(F)(X)=hk'$. To see that it is indeed invertible, recall that if $(G,H,K)\in\Cat{2-matchSub}$, then any element of $HK$ can be uniquely written as an element of $KH$ and therefore correspond to a unique square in $\Gamma{G}$.
\end{proof}

\begin{Lem}Let \Cat{2-Match} be the full subcategory of \Cat{2-matchSub} such that $(G,H,K)\in\Cat{2-Match} \iff G=HK$. Then $\Lambda$'s image is in \Cat{2-Match} and the above adjunction is an equivalence of categories.
\end{Lem}

\begin{proof} We have already shown that the image is what is claimed, but to show that the adjunction is an equivalence, we need to show that its unit and counit are isomorphisms. Let $\nu:\tau\to\Gamma(\Lambda(\tau))$ and $\eta:\Lambda(\Gamma(G))\to G$ be the adjunction's unit and counit. Then they are given by
\begin{align*}
	\nu(X)=(\begin{bmatrix}\imath&X\\\imath&\imath\end{bmatrix}=\begin{bmatrix}\imath&\imath\\X&\imath\end{bmatrix})\\
	\eta(hk=kh)=hk
\end{align*}
and are both iso.
\end{proof}

This proof can also be found is \cite{Andruskiewitsch2009}, and it shows that vacant double groupoids are in correspondence with factorizations of groupoids. In the group case it means that vacant double groups are in correspondence with bi-crossed products of groups, matched pairs of groups or Zappa-Szep products of groups, as all appear in the literature. We now need to remove the slimness condition.

\begin{Def}Let \Cat{MaxExcl} be the category whose objects are sections of maximal exclusive double groupoids over their frame and whose objects are section preserving functors. Let \Cat{2-semi} be the following category :
\begin{itemize}
	\item objects are quadruples $(G,H,K,A)$, where $(G,H,K)\in\Cat{2-matchSub}$ and A an abelian subgroupoid of G such that $hah^{-1}\in A$, $kak^{-1}\in A$ and $A\cap H=A\cap K$ is discrete.
	\item arrows are functors sending the subgroupoids into their counterparts.
\end{itemize}
\end{Def}

Given such a quadruple, one can build a double groupoid with a canonical section the following way : squares are pairs $(hk'=kh',a)$ and the section is $(hk'=kh')\to(hk'=kh',e)$. Composition is given by 
\begin{align*}
	\bigl(hk'=kh',a\bigr)\circ_h\bigl(h''k''=k'h''',b\bigr)&=\bigl((hh'')k''=k(h'k'''),ahbh^{-1}\bigr)\\
	\bigl(hk'=kh',a\bigr)\circ_v\bigl(h'k'''=k''h'',b\bigr)&=\bigl(h(k'k''')=(kk'')h''),akbk^{-1}\bigr)
\end{align*}
Identities are $(h\imath(b)=\imath(a)h,\imath(\imath{a}))$ and $(\imath(a)k=k\imath(b),\imath(\imath(a)))$ and associativity boils down to 
\begin{align*}
hbh^{-1}(hh'c(hh')^{-1})&=hbh'ch'^{-1}h^{-1}\\
&=h(bh'ch'^{-1})h^{-1}
\end{align*}
and inverses are given by
\begin{align*}
	(hk'=kh',a)^{-h}&=(h^{-1}k=k'h'^{-1},h^{-1}ah)\\
	(hk'=kh',a)^{-v}&=(h'k'^{-1}=k^{-1}h,k^{-1}ak)\\
\end{align*}

An arrow $f\in\Cat{2-semi}$ then defines a double functor by $f(hk'=kh',a)=(f(h)f(k')=f(k)f(h'),f(a))$.

\begin{Lem}Let $\Gamma:\Cat{2-semi}\to\Cat{MaxExcl}$ be the functor defined above, then there exist a functor $\Lambda$ and an isomorphism $\phi$ such that $(\Lambda,\Gamma,\phi)$ is an adjunction. Moreover, restricted to the full subcategory of \Cat{2-semi} whose objects are such that $G=AKH$, the adjunction is an equivalence of categories.
\end{Lem}

\begin{proof}Given a section $!:\,\blacksquare \tau\to \tau$ of a double groupoid over its frame, maximality and exclusivity ensure that, for given squares X and Y, there exist a unique square with subdivision :
\begin{align*}
\begin{bmatrix}X& !\\!&Y\end{bmatrix}
\end{align*}
Just as before, this defines a composition on the set of squares of $\tau$. When $!$ is a double functor this composition is associative and unital. Moreover it then has inverses given by 
\begin{align*}X^{-1}:=\imath_v(t_h(X)^{-v}) X^{-h} \imath_v(s_h(X)^{-v})=\imath_h(t_v(X)^{-h}) X^{-v} \imath_h(s_v(X)^{-h})
\end{align*}
and hence defines a groupoid.\\
Since arrows in \Cat{MaxExcl} preserve sections, it preserves composition of the above defined groupoid and the process defines a functor $\Lambda$ as claimed. Now 
\begin{align*}
	\phi:\Cat{MaxExcl}(\tau,\Gamma(G))\to\Cat{2-semi}(\Lambda(\tau),G)
\end{align*}
 is given by $\phi(F)(X)=ahk'$ if $F(X)=(hk'=kh',a)$. To see that it is invertible, we need to recall a previous lemma that shows that $X=a(!(\Pi(X)))$ in $\Lambda(\tau)$ since $X$ and $!(\Pi(X))$ have the same boundary.Therefore $f\in\Cat{2-semi}(\Lambda(\tau),G)\iff f(X)\in AHK\subset G$, in which case corresond to a unique square in $\Gamma(G)$, proving the adjunction claim.
Now suppose that $G=AHK$, then maps in $\Cat{2-semi}(G,\Lambda(\tau))$ give uniquely maps in $\Cat{MaxExcl}(\Gamma(G),\tau)$ by $(ahk')\to X$ being sent to $(hk'=kh',a)\to X$. This finishes the proof.
\end{proof}

\bibliography{C:/SeveM/__Science/__Thesis/References}{}
\bibliographystyle{plain}

\end{document}